\documentclass[11pt]{article}

\usepackage[utf8]{inputenc}
\usepackage[T1]{fontenc}
\usepackage{mathptmx}
\usepackage[margin=1in]{geometry}
\usepackage{amsmath,amssymb,amsthm,mathtools}
\usepackage{graphicx}
\usepackage{setspace}
\usepackage{booktabs}
\usepackage[backend=biber, style=authoryear, maxcitenames=4, maxbibnames=99]{biblatex}
\usepackage[hidelinks]{hyperref}
\usepackage{titlesec}
\usepackage{enumitem}
\usepackage{microtype}

\hypersetup{
    colorlinks=true,
    linkcolor=blue,
    filecolor=magenta, 
    citecolor=blue,
    urlcolor=blue,
    pdftitle={Overleaf Example},
    pdfpagemode=FullScreen,
}

\titleformat{\section}{\normalfont\large\bfseries}{\thesection.}{0.5em}{}
\titleformat{\subsection}{\normalfont\normalsize\bfseries}{\thesubsection}{0.5em}{}

\setlist[itemize]{leftmargin=2em,itemsep=0pt,topsep=2pt}
\setlist[enumerate]{leftmargin=2em,itemsep=0pt,topsep=2pt}

\newtheorem{proposition}{Proposition}
\newtheorem{lemma}{Lemma}
\newtheorem{corollary}{Corollary}

\newtheorem{assumption}{Assumption}

\newcommand{\E}{\mathbb{E}}
\newcommand{\R}{\mathbb{R}}
\newcommand{\pdet}{p_{\mathrm{det}}}

\title{\large\bfseries Buying the Right to Monitor:\\
Editorial Design in AI-Assisted Peer Review}
\author{Zaruhi Hakobyan\thanks{University of Luxembourg, Email: zaruhi.hakobyan@uni.lu}}
\date{\today}

\begin{document}
\maketitle

\begin{abstract}
\noindent Generative AI acts as a disruptive technological shock to evaluative organizations. In academic peer review, it enters both sides of the market: authors use AI to polish submissions, and reviewers use it to generate plausible reports without exerting evaluative effort. We develop a three-sided equilibrium model to analyze this dual adoption and derive a counterintuitive managerial implication for journal policy. We show that when AI capability crosses a critical threshold, reviewer effort collapses discontinuously. This transition creates a welfare misalignment: authors benefit from a weakened ``rat race,'' while editors suffer from degraded signal informativeness. Characterizing the editor's optimal constrained response, we identify a strict policy reversal. Before the AI transition, editors should tighten acceptance standards to curb rent-dissipating author polishing. After the transition, conventional intuition fails: editors must loosen acceptance standards while investing in AI detection, because further tightening only amplifies dissipative polishing without improving sorting. We prove analytically that this sign reversal is a structural consequence of the reviewer effort collapse under log-concave quality distributions. Ultimately, addressing AI in evaluative systems requires treating monitoring and loosened selectivity as complementary design instruments.

\vspace{0.5em}
\end{abstract}

\section{Introduction}
\label{sec:intro}

Academic peer review is an evaluative organization built on a fragile incentive structure. Authors compete for scarce publication slots by investing in how their work is presented, while reviewers---typically weakly compensated and time-constrained---supply judgment on behalf of journals. Generative AI changes both sides of this system. Authors can use large language models to polish manuscripts, improve exposition, and reduce language frictions. Reviewers can use the same tools to generate plausible-looking reports without conducting a substantive evaluation.

These two uses of AI are often treated as a single problem: the diffusion of a new technology into academic publishing. This paper argues that they are economically distinct. Author-side AI primarily lowers the cost of producing polished and persuasive submissions. Reviewer-side AI lowers the cost of appearing to evaluate. When these two margins interact, the equilibrium response of journals can differ sharply from standard editorial intuition.

Recent evidence suggests that both margins are already empirically relevant. On the author side, \cite{Liu2025Adoption} document rapid growth in AI-assisted scientific writing, with adoption rising approximately 400\% in non-English-speaking countries compared to 183\% in English-speaking countries over 2021--2024. On the reviewer side, \cite{Liang2024Stanford} estimate that a meaningful fraction of recent AI-conference reviews show linguistic signatures consistent with AI generation, and document that AI-using reviewers engage less in author rebuttals.

The central puzzle is therefore not simply whether AI improves or harms peer review. If AI makes review reports easier to produce, journals may appear to gain from a larger pool of willing reviewers. But the marginal reviewer may participate precisely because AI makes low-effort reviewing feasible. Thus, AI can increase apparent participation while reducing the average amount of real evaluation. The journal receives more reports, but less information.

We develop a three-sided model of this interaction. Authors choose how much to invest in manuscript polishing. Potential reviewers, heterogeneous in conscientiousness, decide whether to accept review invitations and whether to exert effort or submit an AI-assisted report. The editor commits to a policy consisting of the number of review invitations, the acceptance rate, and the intensity of detection. Generative AI is summarized by a capability parameter that determines how plausible AI-generated reports appear to the editor.

The model delivers three results. First, author polishing takes the form of a rat race. In symmetric equilibrium, polishing improves an author's relative signal but does not change the aggregate acceptance rate. Authors therefore spend resources to compete against one another, generating private incentives but social dissipation. The intensity of this rat race depends on reviewer effort: when reviewers read carefully, polish has more leverage; when reviewers shirk, polish becomes less effective.

Second, reviewer behavior exhibits a participation transition. Below a critical AI capability, shirking is unattractive, and only sufficiently conscientious reviewers accept invitations. Once AI becomes capable enough to generate plausible reports, shirking becomes privately attractive. The reviewer pool expands, but average evaluative effort falls. This distinction between participation and effort composition is central: more reviewers may accept, but fewer may actually evaluate.

Third, the transition creates a welfare misalignment. Authors benefit from reviewer shirking because the polishing rat race weakens. Editors lose because review signals become less informative. The same technological change can therefore make authors better off while making journals worse off.

We then characterize the editor's constrained response. The editor chooses policy subject to an author-welfare constraint: reforms must not make authors worse off than under the decentralized AI equilibrium. The optimal response exhibits a sign reversal. Before AI-assisted reviewing becomes prevalent, the editor should tighten acceptance, because greater selectivity improves sorting and reduces wasteful polishing. After the reviewer-effort transition, the optimal response reverses: the editor should loosen acceptance and introduce detection. Detection restores some effective reviewer effort, but it also increases the return to author polishing and thereby intensifies the rat race. To make detection acceptable to authors, the editor compensates them through a higher acceptance rate. \emph{The editor buys the right to monitor reviewers by loosening selectivity.}

The main managerial implication is counterintuitive. When AI-assisted reviewing becomes prevalent, journals should not automatically respond by becoming more selective. Tighter selectivity may amplify dissipative author competition without restoring the informativeness of review. A better response combines monitoring of reviewer shirking with less stringent acceptance.

The contribution is fourfold. First, we provide a tractable three-sided equilibrium framework for studying the joint response of authors, reviewers, and editors to generative AI. Second, we derive a counterintuitive policy implication: the direction of optimal editorial adjustment reverses once AI-assisted reviewing becomes prevalent. Third, we establish analytically, under log-concave quality distributions and quadratic polish cost, that the sign-flipped author-welfare condition driving the reversal follows from the standard rat-race condition combined with a measurable sharpness condition on the participation transition; this converts what was previously a separate assumption into a structural consequence and isolates the precise way in which the reviewer-side phase transition propagates into editorial design. Fourth, we identify an organizational-design principle that extends beyond peer review: when a technology simultaneously lowers the cost of costly signaling and reduces the informativeness of evaluation, the optimal institutional response combines monitoring with \emph{less} stringent selection.

The paper proceeds as follows. Section \ref{sec:litreview} situates the contribution. Section \ref{sec:model} presents the model. Section \ref{sec:participation} analyzes reviewer participation. Section \ref{sec:ratrace} studies author polishing. Section \ref{sec:editor} characterizes the editor's decentralized response. Section \ref{sec:policy} derives the sign reversal. Section \ref{sec:compstats} presents comparative statics across calibrations. Section \ref{sec:managerial} translates results into managerial implications. Section \ref{sec:robustness} discusses robustness. Section \ref{sec:conclusion} concludes.

\section{Related Literature}
\label{sec:litreview}

This paper relates to four strands of research: rat races and rent dissipation, reviewer incentives in weakly compensated evaluation, AI-assisted scientific writing and peer review, and information design with monitoring.

\paragraph{Rat races, contests, and rent dissipation.} Our model of author polishing builds on classic work on rat races and socially dissipative competition. \cite{Akerlof1976Rat} and \cite{Landers1996Rat} study environments in which agents overinvest in visible signals because relative performance matters. The contest literature similarly emphasizes that individually rational effort can be socially wasteful when agents compete for fixed prizes \cite{Tullock1980Rent, Konrad2009Contests}. The author side of our model follows this logic: manuscript polishing raises an author's relative chance of acceptance but does not change the aggregate number of accepted papers. In symmetric equilibrium, polishing is therefore rent-dissipating. Our contribution is to make the return to this rat race depend on evaluator behavior. When reviewers exert effort, polished presentation affects informative signals; when reviewers shirk, polish has less leverage. Thus, evaluator effort determines the intensity of the author contest---a feature absent from standard contest models, where the technology mapping effort into success is fixed.

\paragraph{Reviewer incentives and weakly compensated evaluation.} A second strand studies peer review and other forms of weakly compensated evaluation. Reviewers often receive little direct monetary compensation, so their effort depends on professional norms, reputational incentives, institutional design, and the perceived cost of reviewing. \cite{Chetty2014} show experimentally that referee behavior responds to institutional incentives even in the absence of conventional market wages. Generative AI changes this margin by lowering the cost of producing a review-like output. Our paper adds a new reviewer incentive margin: the ability to produce a plausible report without substantive evaluation. AI changes not only the cost of reviewing but also the outside option of shirking while appearing to comply, separating participation from effort composition---a distinction central to the model's phase transition and to the welfare conflict between authors and editors.

\paragraph{AI-assisted scientific writing and peer review.} A growing literature studies how generative AI affects scientific writing and reviewing. \cite{Liu2025Adoption} provide large-scale evidence that AI-assisted writing has grown rapidly in biomedical publishing, especially among non-native English-speaking and less-established researchers, with a modest associated productivity gain consistent with our modeling choice to treat author AI primarily as polish rather than latent-quality enhancement. \cite{Lepp2025Bias} show that language-related biases in scientific publishing persist even when AI tools are widely available, suggesting linguistic frictions are not fully eliminated by AI adoption. On the reviewer side, \cite{Liang2024Stanford} provide corpus-level evidence of AI adoption in academic conference reviews and document that AI-using reviewers engage less in author rebuttals---consistent with our treatment of AI-assisted reviewing as an extensive-margin effort collapse. Our paper differs from this literature by modeling AI-assisted reviewing as an endogenous effort choice in equilibrium, rather than asking only whether AI-generated reviews are useful or detectable.

\paragraph{Information design and monitoring in organizations.} The editor's problem connects to information design and monitoring in organizations. \cite{BergemannMorris2019} provide a general framework in which a designer chooses information structures to influence equilibrium behavior. \cite{Lu2024Peer} study mechanisms for eliciting informative text evaluations in the presence of large language models. \cite{Frederiksen2020} and \cite{Cai2023} provide related frameworks for monitoring and effort in organizations. Our setting differs from this literature in two ways. First, the editor does not control only the information structure; she also faces strategic responses from authors and reviewers. Second, the same technology affects both sides of the market: it lowers the cost of author polish and lowers the cost of reviewer shirking. The editor's optimal policy must therefore account for both sorting quality and rent dissipation. This interaction generates the paper's main result: the direction of optimal editorial reform reverses after reviewer AI becomes prevalent. To our knowledge, this sign reversal has not been identified in prior work on peer review, AI adoption, or information design.

\section{Model}
\label{sec:model}

\paragraph{Authors.} A unit mass of authors submits to a journal. Each author will eventually have latent quality $\theta \sim F$ with continuous density $f$ on $[\underline\theta, \overline\theta]$. Each author chooses polishing intensity $a \in [0, \infty)$ at cost $c_A(a)$ \emph{before} $\theta$ is realized. This captures the timing of real manuscript preparation, in which presentation-level investments---language editing, figure polish, formatting---are made during drafting, prior to the author's full certainty about how referees will assess latent quality. In particular, structural formatting, language editing, and AI-assisted drafting are sunk costs incurred prior to submission, long before peer review reveals the paper's standing relative to the competition. Polishing affects presentation without altering latent quality. Publication value is $V > 0$.

\begin{assumption}\label{ass:authors}
(i) $f$ is continuous, strictly positive, and log-concave on $[\underline\theta, \overline\theta]$. (ii) $c_A : [0, \infty) \to \R_+$ is twice continuously differentiable, strictly convex, satisfies $c_A(0) = c_A'(0) = 0$, and $c_A'(a) \to \infty$ as $a \to \infty$.
\end{assumption}

Log-concavity of $f$ is a standard assumption satisfied by most quality distributions of interest (uniform, normal, log-normal with parameter restrictions, beta with shape parameters $\geq 1$). It guarantees that the aggregate signal density $h_m$ is also log-concave (since convolutions preserve log-concavity), which delivers monotone comparative statics in the editor's problem and allows us to derive analytically the sign-flip condition that drives our central reform result (Section \ref{sec:policy}).

Given editorial policy $\pi$, an author's ex-ante payoff is
\begin{equation}
U_A(a; \pi) = V \cdot \E_\theta[\Pr(\text{accepted} \mid \theta, a; \pi)] - c_A(a).
\label{eq:author-payoff}
\end{equation}
This ex-ante specification yields a symmetric-polish equilibrium in which $a$ does not depend on $\theta$; allowing polish to condition on privately observed $\theta$ would produce a separating equilibrium with richer sorting implications, which we leave to future work.

\paragraph{Reviewers.} There is a unit mass of potential reviewers, each with conscientiousness type $t \in [0, 1]$ drawn from $G$ with density $g$. Higher $t$ means lower cost of exerting effort. Reviewers decide whether to accept an invitation before observing the paper; conditional on accepting, they choose effort $e \in \{0, 1\}$ (read and evaluate vs.~produce AI report). Declining yields zero.

Generative AI capability $\gamma \in [0, 1]$ affects the appearance of a shirking report:
\begin{equation}
\alpha(e, \gamma) = e + \gamma(1 - e).
\label{eq:appearance}
\end{equation}
The reviewer receives an appearance reward $\psi_\alpha > 0$. Accepting has baseline payoff $R \in \R$. Real effort costs $c_R(t)$. The editor can choose detection $\pdet \in [0, 1)$; detected shirking reports receive no appearance reward and incur penalty $\ell \geq 0$. Payoffs:
\begin{align}
U_R^{\text{work}}(t) &= R + \psi_\alpha - c_R(t), \label{eq:work-payoff}\\
U_R^{\text{shirk}}(\gamma, \pdet) &= R + \psi_\alpha \gamma (1 - \pdet) - \ell \pdet, \label{eq:shirk-payoff}\\
U_R^{\text{decline}} &= 0. \nonumber
\end{align}

\begin{assumption}\label{ass:reviewers}
(i) $G$ has continuous, strictly positive density $g$ on $[0, 1]$. (ii) $c_R : [0, 1] \to \R_+$ is continuous and strictly decreasing. (iii) $R < 0 < R + \psi_\alpha$ and $c_R(1) < R + \psi_\alpha < c_R(0)$.
\end{assumption}

The parameter $R$ captures the net baseline cost of accepting a review invitation before accounting for any appearance reward or effort cost: it represents time commitment, deadline obligations, and system-management overhead that accepting imposes. Assumption~\ref{ass:reviewers}(iii) requires $R < 0$ so that, absent the appearance reward, accepting is net-costly—consistent with the standard modeling of voluntary unpaid evaluation tasks. The condition $R + \psi_\alpha > 0$ then ensures that effortful reviewing is nonetheless worthwhile for sufficiently conscientious types.

\paragraph{Signals.} Reviewer $i$'s report is
\begin{equation}
r_i(\theta, a, e_i) = e_i (\theta + \beta a) + \nu_i(e_i),
\label{eq:single-signal}
\end{equation}
with $\beta > 0$. When $e_i = 1$ the report is informative; when $e_i = 0$ it is pure noise.

\begin{assumption}\label{ass:signals}
$\nu_i(e_i) \sim \mathcal{N}(0, \sigma_e^2)$ if $e_i = 1$ and $\mathcal{N}(0, \sigma_s^2)$ if $e_i = 0$, with $\sigma_s > \sigma_e > 0$. Noises are independent across reviewers.
\end{assumption}

If $N$ reviewers are invited and fraction $m$ of the accepting pool exerts effort, the aggregate signal $\bar s = N^{-1} \sum_i r_i$ satisfies, taking the effort composition as deterministic (an approximation exact in the large-$N$ limit and accurate for $N \geq 2$),
\begin{equation}
\bar s \mid \theta \sim \mathcal{N}\!\left( m(\theta + \beta a), \, \frac{\Sigma^2(m)}{N} \right), \quad \Sigma^2(m) = m \sigma_e^2 + (1 - m) \sigma_s^2.
\label{eq:aggregate-signal}
\end{equation}

Detection discards each shirking report with probability $\pdet$. We treat detection as a strictly \emph{ex-post screen}: reviewer participation and effort decisions in Section \ref{sec:participation} are made before the editor implements detection, and detection acts only on submitted reports without altering the reviewer's incentive calculus. Equivalently, reviewers face $\pdet = 0$ when deciding whether to accept and how much effort to exert, while the editor applies the screen $\pdet > 0$ to the resulting report pool.\footnote{An alternative specification in which $\pdet$ enters the reviewer's payoff (so that reviewers anticipate detection and adjust effort) would generate an additional behavioral channel through which detection affects $m$. We focus on the ex-post specification because it isolates the editorial-design mechanism that drives our central result and matches the institutional reality that journal-side detection is implemented after reviewers have submitted.} Under this convention, two consequences follow. First, the \emph{effective effort share} among retained reports rises:
\begin{equation}
M(m, \pdet) = \frac{m}{m + (1 - m)(1 - \pdet)}.
\label{eq:effective-effort}
\end{equation}
Second, the \emph{retained sample size} falls: only $N_{\mathrm{ret}}(m, \pdet) = N[m + (1-m)(1-\pdet)] \leq N$ reports survive, with strict inequality whenever $\pdet > 0$ and $m < 1$. The aggregate retained signal satisfies
\begin{equation}
\bar s_{\mathrm{ret}} \mid \theta \sim \mathcal{N}\!\left( M(\theta + \beta a), \, \frac{\Sigma^2(M)}{N_{\mathrm{ret}}(m, \pdet)} \right).
\label{eq:aggregate-signal-detection}
\end{equation}
Detection therefore trades composition quality (higher $M$) against sample size (lower $N_{\mathrm{ret}}$); the net effect on signal precision is positive whenever $\sigma_s > \sigma_e$ strictly, but the sample-size effect attenuates the marginal value of detection and enters the editor's optimization directly.

\begin{assumption}\label{ass:marginal}
For each $(K, N)$, the function $\Lambda(m, K, N, a) = m h_m(\tau_K; a)$ is continuous in $(m, a)$ and weakly increasing in $m$. The polishing cost $c_A$ is strictly convex (Assumption \ref{ass:authors}), which is sufficient for a unique interior solution to the equation $c_A'(a) = V \beta \Lambda(m, K, N, a)$.
\end{assumption}

\paragraph{Editor.} The editor commits to $\pi = (N, K, \pdet)$ with $N \in \mathbb{N}$, $K \in (0, 1)$, and $\pdet \in [0, 1)$, and accepts the top $K$ fraction of papers by aggregate signal. Per-paper invitation cost is $\epsilon N$; detection cost is $D(\pdet)$.

\begin{assumption}\label{ass:detection}
$D : [0, 1) \to \R_+$ is twice continuously differentiable, strictly convex, with $D(0) = D'(0) = 0$ and $D'(\pdet) \to \infty$ as $\pdet \to 1$.
\end{assumption}

The editor's payoff is
\begin{equation}
U_E(\pi) = \E[\theta \mid \text{accepted}; \pi] - \epsilon N - D(\pdet).
\label{eq:editor-payoff}
\end{equation}

\paragraph{Timing and equilibrium.} The editor announces $\pi$; reviewers decide participation and effort; authors observe $\pi$ and $m$ and choose polish; signals realize and the editor applies detection and accepts. The equilibrium concept is subgame-perfect Nash with symmetric author strategies.

\section{Reviewer Participation}
\label{sec:participation}

The reviewer's problem is a three-way choice among decline, shirk, and work. As generative AI capability $\gamma$ rises, the payoff from shirking rises, and at a critical threshold shirking switches from dominated-by-declining to weakly-dominating-declining. Above this threshold, low-conscientiousness reviewers who previously declined now accept-and-shirk; the reviewer pool expands, but average effort falls.

\begin{lemma}\label{lem:reviewer}
Under Assumption \ref{ass:reviewers}, define $S(\gamma, \pdet) = R + \psi_\alpha \gamma (1 - \pdet) - \ell \pdet$ and
\[
\widetilde t(\gamma, \pdet) = c_R^{-1}\!\left( \psi_\alpha[1 - \gamma(1 - \pdet)] + \ell \pdet \right),
\]
with the convention $c_R^{-1}(x) = 0$ if $x \geq c_R(0)$ and $c_R^{-1}(x) = 1$ if $x \leq c_R(1)$. If $S(\gamma, \pdet) < 0$: shirking is dominated by declining; only types with $t \geq c_R^{-1}(R + \psi_\alpha) \equiv t_0$ accept and work. If $S(\gamma, \pdet) \geq 0$: all types accept; types $t \geq \widetilde t(\gamma, \pdet)$ work and the rest shirk.
\end{lemma}

\begin{proof} See Appendix \ref{app:reviewer-proof}. \end{proof}

The critical AI capability at which shirking becomes weakly profitable is $\gamma_1 = -R/\psi_\alpha$, with $\gamma_1 \in (0, 1)$ under Assumption \ref{ass:reviewers}.

\begin{proposition}[Participation phase transition]\label{prop:participation}
At $\pdet = 0$:
\begin{enumerate}[label=(\alph*)]
\item For $\gamma < \gamma_1$: pool mass $\mu = 1 - G(t_0)$ and effort rate $m = 1$.
\item For $\gamma \geq \gamma_1$: $\mu = 1$ and $m(\gamma) = 1 - G(\widetilde t(\gamma, 0)) < 1$.
\item At $\gamma = \gamma_1$, $\mu$ jumps up by $G(t_0) > 0$ and $m$ drops by $G(\widetilde t(\gamma_1, 0)) > 0$.
\end{enumerate}
\end{proposition}

\begin{proof} See Appendix \ref{app:participation-proof}. \end{proof}

AI changes \emph{who} reviews and \emph{how much} real effort they exert. This discontinuous shift is the engine of the downstream results.

A note on how these two margins enter the model. Pool mass $\mu$ determines whether the editor can staff a paper with $N$ reviewers: as long as $\mu N \geq 1$ (trivially satisfied in our setting), staffing is feasible and $\mu$ itself does not enter the editor's payoff. Effort composition $m$ enters directly through the aggregate signal \eqref{eq:aggregate-signal} and the selection quality $Q(N, K, M)$. The narrative tension ``more reviewers accept but fewer exert effort'' is therefore a tension between \emph{participation} and \emph{composition}; the formal welfare consequences operate through composition ($m$, or $M$ under detection), while the participation margin matters for the qualitative description of who reviews in each regime and for the shadow pool of AI-willing reviewers the editor draws from.

\section{Author Rat Race}
\label{sec:ratrace}

Fix reviewer effort rate $m \in (0, 1]$. Authors take $m$ as given. The private return to polish scales with $m\beta$: polish affects only effortful reviewers' signals, so when reviewers shirk more, polish has less leverage.

\begin{proposition}[Author rat race]\label{prop:ratrace}
Under Assumptions \ref{ass:authors}, \ref{ass:signals}, and \ref{ass:marginal}, for any $m \in (0, 1]$ there exists a unique symmetric equilibrium polish $a^*(m)$ satisfying
\begin{equation}
c_A'(a^*(m)) = V \beta m \, h_m(\tau_K; a^*(m)),
\label{eq:author-foc}
\end{equation}
weakly increasing in $m$. Rat-race dissipation $c_A(a^*(m))$ is weakly increasing in $m$.
\end{proposition}

\begin{proof} See Appendix \ref{app:ratrace-proof}. \end{proof}

In symmetric equilibrium every author has acceptance probability $K$, unchanged from a counterfactual with $a = 0$. Polishing is pure rent dissipation.

\begin{corollary}[Author welfare and reviewer effort]\label{cor:author-welfare}
In symmetric equilibrium, $U_A^*(m) = KV - c_A(a^*(m))$ is weakly decreasing in $m$.
\end{corollary}

Authors benefit when reviewers shirk: the shirking environment weakens the contest.

\section{Editor Response and Welfare Misalignment}
\label{sec:editor}

Let $Q(N_{\mathrm{ret}}, K, M)$ denote expected accepted-paper quality, where $N_{\mathrm{ret}} = N[m + (1-m)(1-\pdet)]$ is the retained sample size and $M = M(m, \pdet)$. The editor's payoff is $U_E = Q(N_{\mathrm{ret}}, K, M) - \epsilon N - D(\pdet)$. Note that the editor pays the invitation cost $\epsilon N$ on the full panel, while sorting quality uses only the retained reports $N_{\mathrm{ret}} \leq N$.

\begin{lemma}[Selection quality]\label{lem:Q}
For fixed $K$ and $N_{\mathrm{ret}}$, $Q$ is strictly increasing in $M$. For fixed $(N_{\mathrm{ret}}, M)$, $Q$ is strictly decreasing in $K$. For fixed $(K, M)$ with $M > 0$, $Q$ is strictly increasing in $N_{\mathrm{ret}}$ with diminishing marginal returns.
\end{lemma}

\begin{proof} See Appendix \ref{app:Q-proof}. \end{proof}

In the decentralized benchmark ($K = K_0$, $\pdet = 0$, so $N_{\mathrm{ret}} = N$), the editor chooses $N^*(\gamma) \in \arg\max_N [Q(N, K_0, m(\gamma)) - \epsilon N]$.

\begin{proposition}[Bounded limits of review intensity]\label{prop:editor}
$N^*(\gamma) \in \mathbb{N}$ is constant at $N_0^*$ for $\gamma < \gamma_1$, may rise at $\gamma_1$ (provided the marginal information value of an additional reviewer exceeds $\epsilon$ at $m = m_1$), and decays to the institutional minimum $N^* = 1$ as $\gamma \to 1$.
\end{proposition}

\begin{proof} See Appendix \ref{app:editor-proof}. \end{proof}

\begin{corollary}[Welfare misalignment]\label{cor:misalignment}
At $\gamma = \gamma_1$, per-author welfare $U_A^*$ strictly rises and editor welfare $U_E^*$ strictly falls (under non-uniform $F$).
\end{corollary}

\begin{proof} See Appendix \ref{app:misalignment-proof}. \end{proof}

Authors and editors have strictly opposed private interests across the shirking threshold. Reforms cannot be built on aligned-welfare grounds.

\section{Editorial Design: The Sign Reversal}
\label{sec:policy}

The editor's constrained reform problem is
\begin{equation}
\max_{N, K, \pdet} U_E(N, K, \pdet) \quad \text{s.t.} \quad U_A(N, K, \pdet) \geq U_A(\pi^{\text{dec}}).
\label{eq:constrained-problem}
\end{equation}

The IR constraint does not follow from submission behavior in our model, where the author mass is fixed at unity: it is a normative constraint the editor respects, capturing the political-economy observation that reforms substantially harming incumbent authors are not sustainable in voluntary-submission journals. Formally, one can microfound this constraint through an outside option (authors submit elsewhere if $U_A$ falls below the decentralized level) or through journal competition; we take the constraint as given here to keep the editorial optimization tractable and return to endogenous submission margins in the concluding discussion.

\paragraph{Treatment of $N$.} For tractability and to focus on the new instruments introduced by AI ($K$ and $\pdet$), we hold $N$ fixed at the decentralized post-transition value $N_0^* \equiv N^*(\gamma_1)$ characterized in Proposition \ref{prop:editor}, and optimize the constrained reform problem over $(K, \pdet)$. This restriction is consistent with institutional reality: journals' per-paper reviewer counts are typically set by editorial conventions and budget constraints rather than reoptimized in response to each technology shift. We verify in Section \ref{sec:compstats} that allowing $N$ to vary modestly does not overturn the sign-reversal result.

The Lagrangian with multiplier $\lambda \geq 0$ is
\begin{equation}
\mathcal{L}(K, \pdet, \lambda) = Q(N_{\mathrm{ret}}(m, \pdet), K, M(m, \pdet)) - D(\pdet) + \lambda \left[ K V - c_A(a^*(K, \pdet)) - \bar{U}_A(\gamma) \right],
\label{eq:lagrangian}
\end{equation}
where the dependence of $Q$ on $\pdet$ runs through both the effective effort share $M$ \emph{and} the retained sample size $N_{\mathrm{ret}}$. The first-order conditions are
\begin{align}
\frac{\partial \mathcal{L}}{\partial \pdet} &= \left(\frac{\partial Q}{\partial M} \cdot \frac{\partial M}{\partial \pdet} + \frac{\partial Q}{\partial N_{\mathrm{ret}}} \cdot \frac{\partial N_{\mathrm{ret}}}{\partial \pdet}\right) - D'(\pdet) - \lambda \, c_A'(a^*) \frac{\partial a^*}{\partial \pdet} = 0, \label{eq:foc-pdet-main}\\
\frac{\partial \mathcal{L}}{\partial K} &= \frac{\partial Q}{\partial K} + \lambda \left[ V - c_A'(a^*) \frac{\partial a^*}{\partial K} \right] = 0. \label{eq:foc-K-main}
\end{align}
The bracketed term in \eqref{eq:foc-pdet-main} is the \emph{net marginal sorting benefit} of detection: the composition gain $(\partial Q/\partial M)(\partial M/\partial \pdet) > 0$ minus the sample-size loss $|(\partial Q/\partial N_{\mathrm{ret}})(\partial N_{\mathrm{ret}}/\partial \pdet)| > 0$ (since $\partial N_{\mathrm{ret}}/\partial \pdet = -N(1-m) < 0$). Detection raises composition quality but discards reports; an interior $\pdet^* > 0$ requires the net effect to dominate the IR-compensation term.

\paragraph{Why $K$ rather than $N$ as the compensating instrument.} The reform problem above optimizes over $(K, \pdet)$ holding $N$ at $N_0^*$. A natural question is whether the editor could instead invite more reviewers ($N > N_0^*$) to compensate for AI-induced signal degradation, leaving $K$ unchanged. We argue this alternative does not work because raising $N$ \emph{tightens} rather than relaxes the authors' IR constraint. Increasing the number of reviewers raises the retained sample size $N_{\mathrm{ret}}$, which sharpens the aggregate signal precision $\Sigma^2(M)/N_{\mathrm{ret}}$. A more precise signal increases the marginal informativeness of each unit of polish, raising the equilibrium $a^*(N)$ in the author's FOC; rat-race dissipation $c_A(a^*)$ rises and author welfare falls. We verify this monotonicity numerically at the post-transition baseline of Table~\ref{tab:baseline}: as $N$ rises from $1$ to $20$, equilibrium polish $a^*$ rises from $0.21$ to $0.85$ and author welfare $U_A^*$ falls from $0.29$ to $0.19$. Therefore $N$ moves authors in the wrong direction. Among the available instruments $(N, K, \pdet)$, only $K$ both increases author welfare and is feasible to deploy as a compensating margin: raising $K$ has a direct mechanical effect $\partial U_A/\partial K > 0$ at the post-transition baseline (Premise iii), whereas raising $N$ has a perverse effect $\partial U_A/\partial N < 0$. This makes loosening $K$ the unique policy compensation channel for the IR violation induced by detection.

Under detection, the author's best-response polish is characterized by the FOC $c_A'(a^*) = V\beta M h_M(\tau_K; a^*)$, extending Proposition \ref{prop:ratrace} by replacing $m$ with the effective effort share $M(m, \pdet)$; the density $h_M$ is evaluated in the retained-signal distribution. When $\pdet = 0$, $M = m$ and the Proposition \ref{prop:ratrace} FOC is recovered.

Before stating the central reform result, we establish an analytical lemma showing that one of the regime-shift conditions can be derived from primitives plus a measurable sharpness condition on the participation transition.

\begin{lemma}[Sharpness lemma]\label{lem:sharpness}
Under Assumption \ref{ass:authors} (log-concave $f$), Assumption \ref{ass:signals} (Gaussian noise), quadratic polish cost $c_A(a) = \kappa a^2/2$, and a regularity condition (R) on the score ratio (specified in the proof), suppose the rat-race sensitivity holds at the pre-transition baseline:
\[
V < c_A'(a^*) \, \frac{\partial a^*}{\partial K} \quad \text{at } m = 1, \, \pdet = 0. \tag{i}
\]
Then there exists a sharpness threshold $\bar m \in (0, 1)$, computable from primitives, such that whenever the post-transition effort rate satisfies $m_1 \leq \bar m$, the sign-flipped condition
\[
V > c_A'(a^*) \, \frac{\partial a^*}{\partial K} \quad \text{at } m = m_1, \, \pdet = 0
\]
holds automatically. The regularity condition (R) is satisfied automatically when $f$ has bounded support and bounded density, and we verify it numerically in our calibration.
\end{lemma}

\begin{proof} See Appendix \ref{app:sharpness-proof}. \end{proof}

Lemma \ref{lem:sharpness} converts what was previously assumed into a structural consequence of two ingredients: the rat-race intensity at pre-transition and the magnitude of the reviewer participation drop. Economically: when the participation transition is sharp enough that effort drops well below $m = 1$ at $\gamma_1$, the polish-cost savings term $c_A'(a^*) \partial a^*/\partial K$ --- which scales with $m^2$ --- falls below the mechanical acceptance value $V$, and the sign flip follows. In the baseline calibration, $\bar m \approx 0.51$ and $m_1 = 0.125$, so the condition holds with substantial margin (factor of approximately $4\times$).

\begin{proposition}[Sign reversal in optimal editorial reform]\label{prop:policy}
Under Assumptions \ref{ass:authors}--\ref{ass:detection}, suppose:
\begin{enumerate}[label=(\roman*)]
\item At the pre-transition baseline ($m = 1$, $\pdet = 0$), the rat-race sensitivity satisfies $V < c_A'(a^*) \, \partial a^* / \partial K$.
\item At the post-transition baseline ($m = m_1 < 1$, $\pdet = 0$), detection intensifies the rat race: $\partial a^* / \partial \pdet > 0$.
\item The post-transition effort rate satisfies $m_1 \leq \bar m$, where $\bar m$ is the sharpness threshold of Lemma \ref{lem:sharpness}.
\item At the post-transition baseline ($\pdet = 0$), the net marginal sorting benefit of detection exceeds the marginal IR-compensation cost:
\[
\frac{\partial Q}{\partial M} \cdot \frac{\partial M}{\partial \pdet} + \frac{\partial Q}{\partial N_{\mathrm{ret}}} \cdot \frac{\partial N_{\mathrm{ret}}}{\partial \pdet} > \lambda \cdot c_A'(a^*) \cdot \frac{\partial a^*}{\partial \pdet},
\]
where $\lambda > 0$ is the IR multiplier at the joint optimum. The first term on the LHS is the composition gain (positive); the second is the sample-size loss (negative); the net sum is the relevant marginal sorting benefit.
\end{enumerate}
Then the constrained optimal policy $(K^*, \pdet^*)$ exhibits:
\begin{enumerate}[label=(\alph*)]
\item Pre-transition ($\gamma < \gamma_1$): $\pdet^* = 0$ and $K^* < K_0$.
\item Post-transition ($\gamma \geq \gamma_1$): $\pdet^* > 0$ and $K^* > K_0$.
\end{enumerate}
The sign of the acceptance-rate adjustment reverses across $\gamma_1$. All four premises are verified in the baseline calibration; Premise (iv) holds with margin throughout the post-transition range.
\end{proposition}

\begin{proof} See Appendix \ref{app:policy-proof}. \end{proof}

The proposition's premises now have a cleaner economic interpretation. Premise (i) is the standard rat-race condition at pre-transition. Premise (ii) is the detection-intensification condition at post-transition. Premise (iii) is a \emph{measurable sharpness condition} on how far reviewer effort drops at the phase transition --- this replaces the previously-assumed sign-flip condition, which now follows from Lemma \ref{lem:sharpness}. Premise (iv) is the interior-detection condition. Among the four, only Premise (iv) remains an assumption that does not follow from more primitive structure; the rest reduce to (i), (ii), and a sharpness threshold $m_1 \leq \bar m$ that can be checked directly given the participation analysis.

Together the premises tell the story of the reform. Pre-transition, tightening the acceptance rate simultaneously improves sorting and relaxes author welfare (via Premise i), and detection has no value (reviewers already work). Post-transition, the editor has a strict incentive to introduce detection because $D'(0) = 0$ and detection improves the effective effort share $M$; by Premise (iv), the sorting benefit exceeds the compensation cost, so the reform remains interior. But detection raises $M$, which by Premise (ii) intensifies the rat race, violating the authors' IR constraint. Retaining detection requires offsetting the IR violation on a separate margin, and $K$ is the only available instrument. By Lemma \ref{lem:sharpness} and Premise (iii), raising $K$ raises $U_A$ post-transition. \emph{The editor buys the right to monitor reviewers by loosening selectivity.}

\begin{corollary}[Incomplete restoration]\label{cor:incomplete-restoration}
For $\gamma \geq \gamma_1$, no policy in $(N, K, \pdet)$ restores pre-AI editor welfare.
\end{corollary}

\begin{proof} See Appendix \ref{app:restoration-proof}. \end{proof}

Full restoration would require instruments outside our policy space: reviewer payments, enforceable AI-use bans, or institutional changes to the reputational system.

\section{Comparative Statics}
\label{sec:compstats}

The sign reversal in Proposition \ref{prop:policy} is derived under four premises whose economic content depends on primitives. We now explore how the optimal policy varies across plausible calibrations, what changes in direction, and what remains invariant. Throughout this section we use the baseline calibration in Table~\ref{tab:baseline}, varying one parameter at a time.

\begin{table}[t]
\centering
\caption{Baseline calibration for numerical comparative statics}
\label{tab:baseline}
\begin{tabular}{lll}
\toprule
Parameter & Value & Interpretation \\
\midrule
$K_0$ & 0.30 & Baseline acceptance rate \\
$\beta$ & 0.50 & Effect of polish on effortful reviews \\
$\kappa$ & 0.30 & Author polishing cost, $c_A(a) = \kappa a^2/2$ \\
$\psi_\alpha$ & 0.25 & Appearance reward for reviewers \\
$R$ & $-0.08$ & Baseline cost of accepting a review \\
$\ell$ & 0.05 & Reputational penalty for detected shirking \\
$\sigma_e$ & 0.30 & Noise under effortful review \\
$\sigma_s$ & 0.40 & Noise under shirking review \\
$\epsilon$ & 0.02 & Editorial cost per reviewer invited \\
$V$ & 1.00 & Publication value \\
$N$ & 2 & Reviewers per paper \\
\bottomrule
\end{tabular}
\end{table}

At the baseline, $\gamma_1 = -R/\psi_\alpha = 0.32$. We solve the constrained editor's problem \eqref{eq:constrained-problem} numerically on a grid of $(K, \pdet)$, using a Gaussian approximation to the threshold density $h_M(\tau_K)$---which governs the author FOC under detection, since polishing is effective through the retained-signal composition $M(m, \pdet)$---and a signal-to-noise posterior formula for $Q$. When $\pdet = 0$, $M = m$ and $h_M$ reduces to $h_m$ of Proposition \ref{prop:ratrace}. The focus is on qualitative patterns across parameter perturbations; exact numerical magnitudes are approximate.

\subsection{Baseline reviewer cost ($R$)}

Figure \ref{fig:cs-R} varies the baseline reviewer cost $R$ from $-0.04$ (low reluctance) to $-0.22$ (high reluctance). Since $\gamma_1 = -R/\psi_\alpha$, the transition threshold shifts rightward as $R$ becomes more negative: $\gamma_1 = 0.16, 0.32, 0.60, 0.88$ for the four calibrations.

\begin{figure}[t]
\centering
\includegraphics[width=\textwidth]{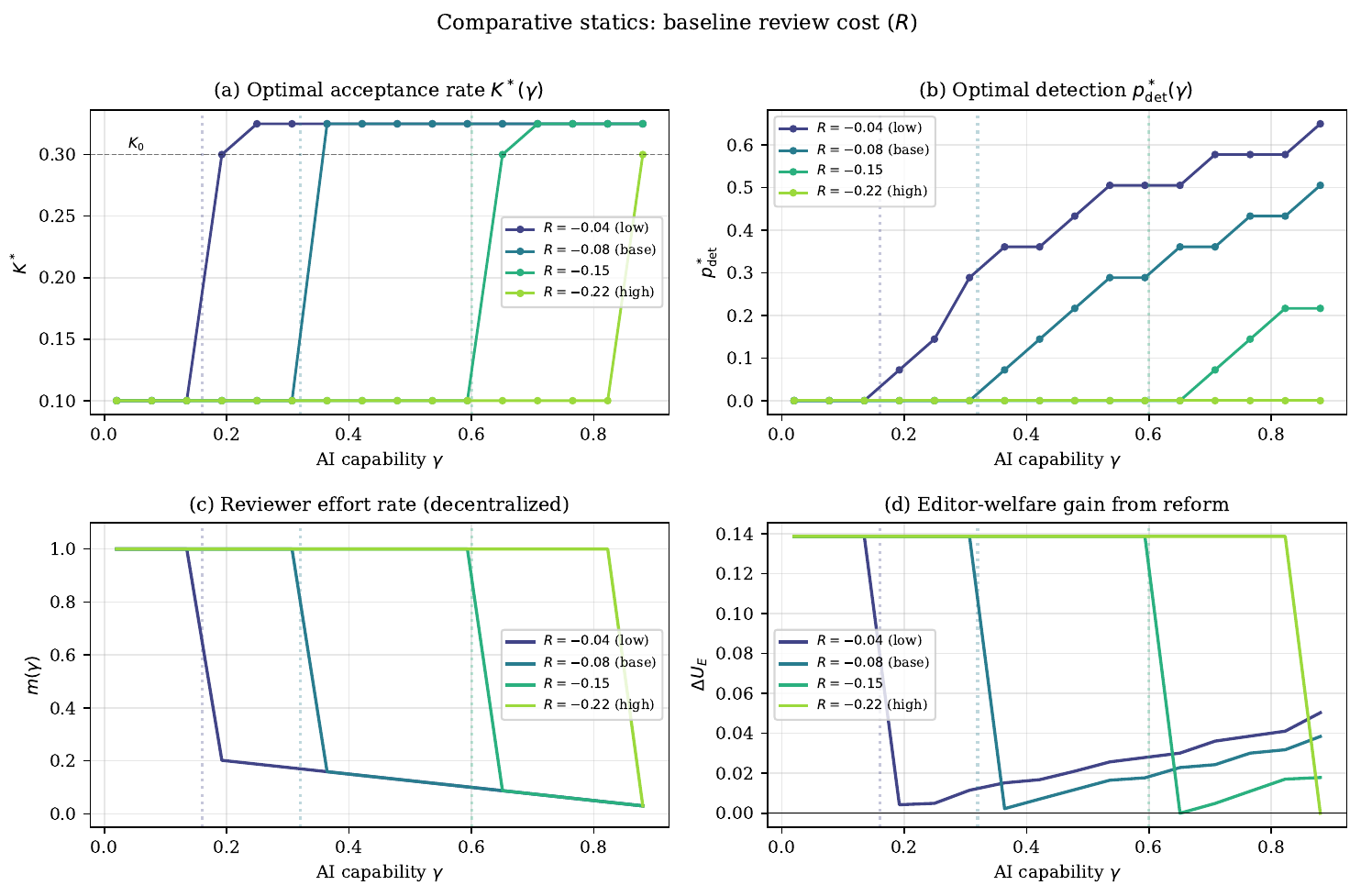}
\caption{Comparative statics in the baseline reviewer cost $R$. Panel (a): the sign reversal in $K^*$ occurs at each calibration's $\gamma_1$. Panel (b): detection intensity is strongest when the transition occurs early. Panel (c): the reviewer-effort drop is vertical at the transition in all calibrations. Panel (d): editor-welfare gain from reform is concentrated pre-transition.}
\label{fig:cs-R}
\end{figure}

Three patterns emerge. First, panel (a) shows the sign reversal occurs in every calibration, at exactly the $\gamma_1$ predicted by the model: $K^*$ jumps from a low pre-transition value to an above-$K_0$ value at each threshold. The direction of reform does not depend on where $R$ sits within the admissible range. Second, panel (b) shows that post-transition detection intensity $\pdet^*$ is strongest when the transition occurs early: at $R = -0.04$, $\pdet^*$ rises to nearly 0.65 as $\gamma \to 0.88$; at $R = -0.15$, it rises only to about 0.22 over the same range. The reason is that late-transitioning calibrations have less ``post-transition runway'' before reviewer effort approaches zero, so the editor's window for productive detection is compressed. Third, panel (d) shows that editor welfare gains from reform are concentrated pre-transition---where the editor can simultaneously tighten $K$ and improve sorting---and relatively flat post-transition. \emph{The reform is most valuable when it arrives early.}

This comparative static has a direct managerial reading: \textbf{journals with more reluctant reviewer pools should invest less in AI detection}, because their phase transition happens at higher AI capability where detection effectiveness is bounded.

\subsection{Publication value ($V$)}

Figure \ref{fig:cs-V} varies publication value $V$ from $0.6$ to $2.4$, holding all other parameters at baseline. Here $\gamma_1$ is invariant at 0.32, but rat-race intensity scales with $V$.

\begin{figure}[t]
\centering
\includegraphics[width=\textwidth]{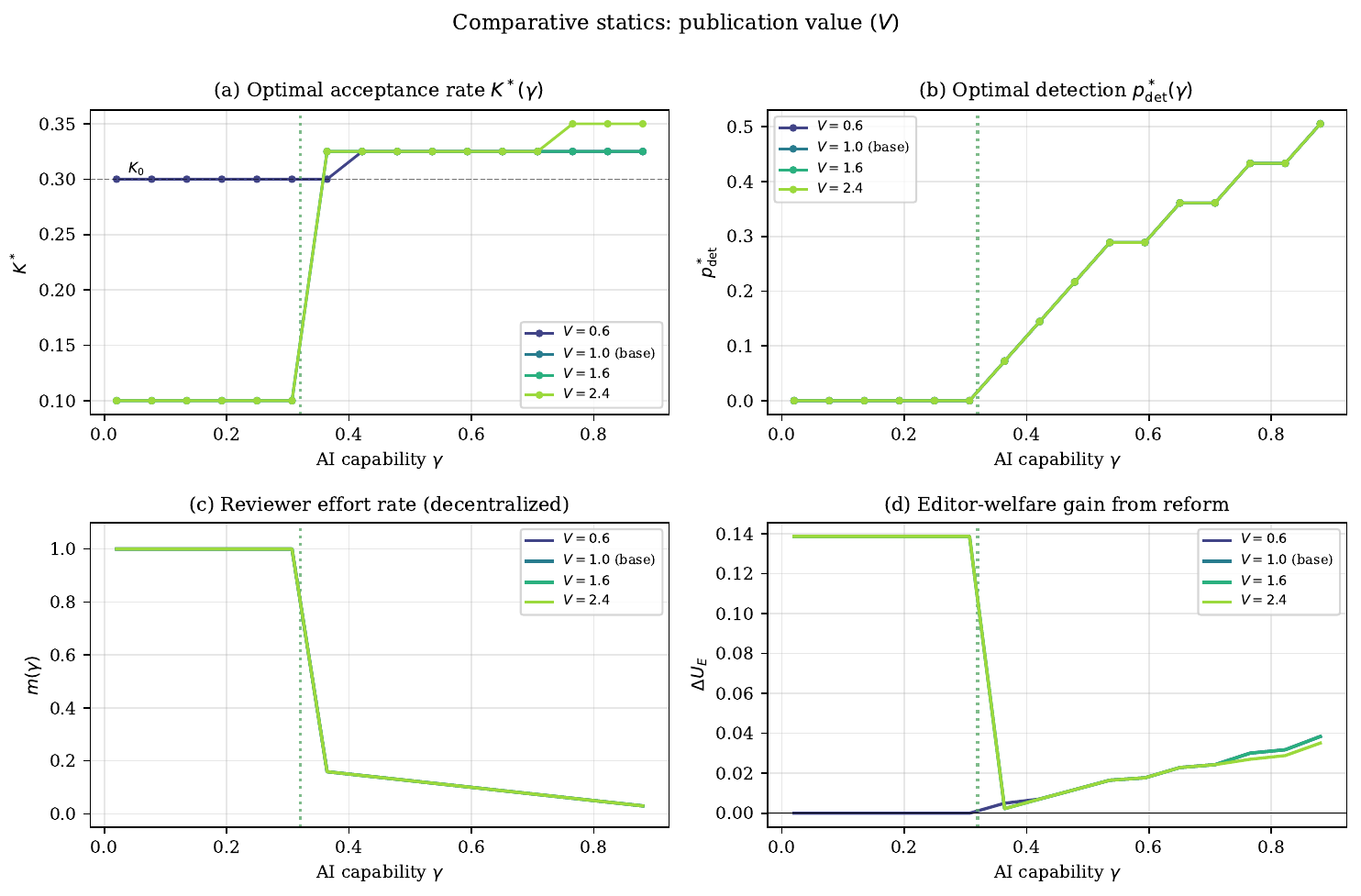}
\caption{Comparative statics in publication value $V$. Panel (a): at $V = 0.6$, the editor does not tighten below $K_0$ pre-transition; at higher $V$, pre-transition tightening is aggressive. Panel (b): detection intensity is largely invariant to $V$. Panel (c): effort rate does not depend on $V$. Panel (d): editor welfare gain is roughly proportional to $V$.}
\label{fig:cs-V}
\end{figure}

Two findings. First, panel (a) reveals that the pre-transition tightening result is not automatic. At $V = 0.6$, the editor sets $K^* = K_0$ pre-transition---the IR constraint binds at the baseline $K_0$ itself and there is no room to tighten. At $V \geq 1.0$, the IR constraint is slack at $K_0$ (the rat race is so intense that authors are willing to accept tighter selectivity in exchange for lower equilibrium polishing), and the editor tightens aggressively. \emph{Pre-transition tightening requires sufficient rat-race intensity; when stakes are low, the editor's optimum is to leave $K_0$ alone.} Second, panel (d) shows that the editor-welfare gain from reform is approximately proportional to $V$: higher-stakes journals benefit more from Pareto-improving reform, in absolute terms.

The magnitude of the reform depends on $V$, but the qualitative direction (tighten pre-transition, loosen post-transition) holds whenever $V$ is large enough for the IR constraint to bind interior.

\subsection{Noise from shirking ($\sigma_s$)}

Figure \ref{fig:cs-sigma} varies the shirking-report noise $\sigma_s$ from 0.33 to 0.75, spanning ratios $\sigma_s/\sigma_e$ from 1.1 to 2.5.

\begin{figure}[t]
\centering
\includegraphics[width=\textwidth]{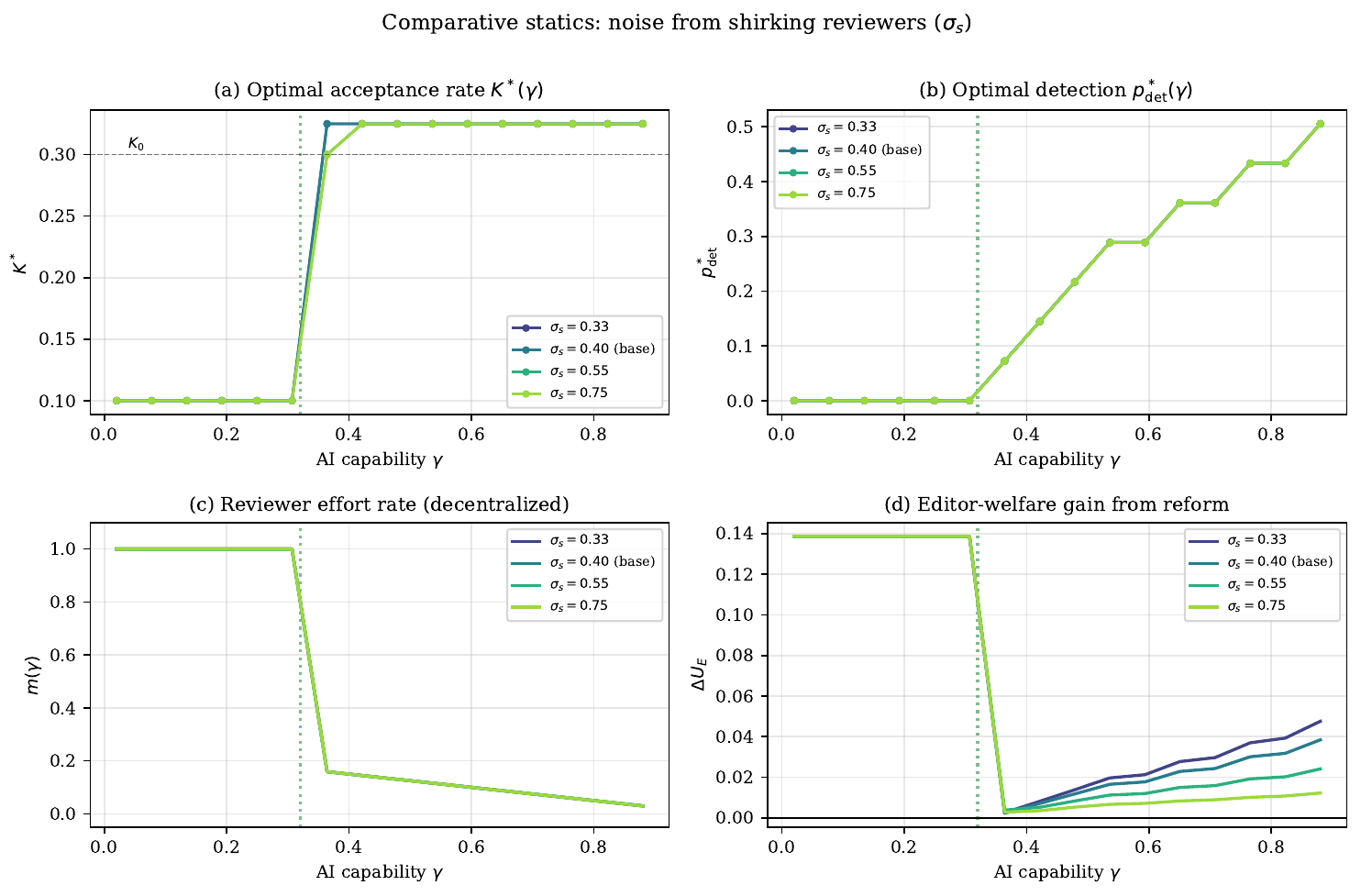}
\caption{Comparative statics in shirking noise $\sigma_s$. Panel (a)--(c): policy and participation responses are largely invariant. Panel (d): editor welfare gain from reform is smaller when shirking is more damaging.}
\label{fig:cs-sigma}
\end{figure}

The policy path $(K^*, \pdet^*)$ is largely invariant to $\sigma_s$ in this calibration. What varies is the editor welfare gain from reform (panel d): when shirking is more damaging ($\sigma_s$ high), reform delivers less welfare recovery. The economic content is that the editor's instruments cannot fully substitute for lost reviewer effort. Detection restores part of the effective effort share $M$, but the remaining signal noise scales with $\sigma_s$, and higher $\sigma_s$ means the post-reform sorting quality is further from the pre-AI frontier. \emph{The reform is a damage-control strategy; its effectiveness declines as the underlying damage from shirking grows.}

\subsection{Summary of comparative statics}

Across the three sweeps, two features of the main result are robust and one is calibration-dependent. The sign reversal itself---$K^*$ crosses $K_0$ at $\gamma_1$---holds in every calibration tested. The direction of the reform is a structural property of the three-sided equilibrium. What varies across calibrations is: the \emph{location} of the transition (shifts with $R$ and $\psi_\alpha$), the \emph{pre-transition tightening magnitude} (grows with $V$), and the \emph{detection intensity} (depends on how much ``post-transition runway'' remains). The welfare gain from reform is largest when the transition occurs early, stakes are high, and shirking is only mildly damaging.

\section{Managerial Implications}
\label{sec:managerial}

\paragraph{1. Author AI and reviewer AI are economically distinct phenomena.} Policies targeting author polishing (e.g., disclosure requirements) address rent dissipation but do not improve sorting. Policies targeting reviewer shirking (e.g., monitoring) improve sorting but do not reduce rent dissipation. Conflating them leads to mismatched policy instruments. An editor worried about paper quality should distinguish whether the concern is signal noise (reviewer side) or dissipative author competition.

\paragraph{2. Selectivity and detection are complements, not substitutes, post-transition.} The classical editorial intuition treats selectivity and monitoring as alternative means to improve quality. Our analysis shows that in an AI-affected review environment, tighter selection alone amplifies rent-dissipating polishing without improving sorting, while detection alone reduces shirking but leaves the rat race intact. The optimal reform is the joint policy: looser selection paired with positive detection. Neither instrument alone dominates the decentralized equilibrium.

\paragraph{3. Editorial adjustment should track observable indicators of reviewer shirking.} Proposition \ref{prop:policy} is keyed to the unobservable $\gamma$, but the relevant observable is the fraction of shirking reviews $(1-m)$. An editor who observes a meaningful fraction of AI-generated reports---estimable through detection trials or text-signature analysis---should consider \emph{raising} acceptance rates rather than lowering them, contra the conventional response. We emphasize that raising $K$ post-transition is not a reduction in quality standards: it is the withdrawal of a rat-race instrument that, in an AI-noisy environment, has stopped delivering sorting value and is only amplifying dissipative author polishing. The editor's quality standard is enforced through sorting (now partly via detection $\pdet$), not through selectivity alone.

\paragraph{4. Investment in AI-detection technology is justified by the sign reversal.} Absent detection, editors face a polarized choice between tightening (pre-AI intuition) and abandoning quality control. Detection creates the intermediate policy that makes the sign-reversal reform feasible. Section \ref{sec:compstats} shows that detection investment pays off most for journals whose transition happens early; late-transitioning journals may not benefit as much.

\paragraph{5. Editorial reform should be framed as Pareto improvement, not restoration.} Corollary \ref{cor:incomplete-restoration} shows reform cannot return journals to the pre-AI welfare frontier. Framing reform as a Pareto-improvement from the decentralized AI equilibrium is likely to attract greater author support, because authors' welfare under the reform is at least as high as under the decentralized equilibrium.

The mechanism extends beyond peer review to evaluative organizations sharing three features: unpaid or weakly compensated agents supply judgment, others invest in costly signals, and a central designer chooses rules. Examples include performance evaluations, grant review, promotion letters, and crowdsourced content moderation. The pattern we identify---tightening standards is counterproductive when technology lowers the cost of superficially adequate evaluations---applies across these settings.

\section{Robustness}
\label{sec:robustness}

Several features of the model deserve explicit discussion.

\paragraph{Author AI as polish versus content.} Our model treats author-side AI adoption as polishing ($a$): improving presentation without changing latent quality. A natural objection is that AI might raise $\theta$ directly---through logical tightening, gap-finding, or error-catching. We take the view that current AI usage is primarily polish-type for two reasons. First, current LLMs do not produce novel substantive contributions; they smooth communication. Second, \cite{Liu2025Adoption} document only a modest productivity gain associated with AI adoption, consistent with polish rather than content enhancement. To the extent that author AI genuinely raises $\theta$, the welfare misalignment we identify is narrower and the managerial recommendation is less urgent; the qualitative sign reversal, however, requires only that author AI includes a polish component.

\paragraph{Functional forms.} The participation threshold $\gamma_1 = -R/\psi_\alpha$ is independent of $F$ and $G$. The rat-race direction holds for any $F$ with continuous positive density and any strictly convex polishing cost $c_A$ satisfying Assumption \ref{ass:authors}. Power-law costs $c_A(a) = \kappa a^\eta/\eta$ for $\eta > 1$ give the same qualitative pattern; only magnitudes change.

\paragraph{Deterministic-composition approximation.} Equation \eqref{eq:aggregate-signal} treats exactly $mN$ reports as effortful and $(1-m)N$ as shirking. For small $N$ this understates the variance of the aggregate signal. The approximation is exact in the large-$N$ limit and is accurate for $N \geq 2$ in simulation; a fully stochastic treatment would enlarge effective $\Sigma^2$ but preserve the direction of all comparative statics.

\paragraph{Binary reviewer effort.} The $\{0, 1\}$ effort choice abstracts from intermediate effort levels. A continuous effort specification with type-dependent precision would soften the phase transition into a smooth adjustment. The discontinuity in participation---new types entering the pool at $\gamma_1$---would remain, and is the mechanism driving the main welfare misalignment. While continuous effort would smooth the transition, the qualitative policy reversal holds whenever AI induces a sufficiently sharp decline in average informativeness---specifically, whenever the post-transition effort rate $m_1$ falls below the sharpness threshold $\bar m$ of Lemma \ref{lem:sharpness}. The discrete formulation produces the cleanest version of this condition; intermediate effort choices would deliver the same result with a smoother dependence of $m_1$ on $\gamma$.

\paragraph{Single polishing intensity.} We treat polish $a$ as scalar. In reality authors may target specific dimensions (grammar, exposition, positioning). A multidimensional polish extension would allow editors to target detection to specific polish components, enabling richer reform designs. We leave this to future work.

\paragraph{Calibration dependence of premises.} Proposition \ref{prop:policy} is conditional on four premises. Premise (i) is the standard rat-race sensitivity at the pre-transition baseline. Premise (ii) is the detection-intensification condition. Premise (iii) is the measurable sharpness condition $m_1 \leq \bar m$ from Lemma \ref{lem:sharpness}, which replaces the previously-assumed sign-flip condition; under log-concave quality distributions, the sign flip is now a structural consequence rather than an independent assumption. Premise (iv) is the interior-detection condition. All four hold at the baseline calibration of Table \ref{tab:baseline} and across the comparative-statics sweeps reported in Section \ref{sec:compstats}. The premises fail in identifiable corners. Premise (i) fails when publication value $V$ is very low (below roughly 0.6 in our setup), so the IR constraint binds at $K_0$ pre-transition and the editor does not tighten; in that corner, the reform reduces to detection alone. Premise (ii) fails when signal noise is so high that the density-concentration effect from raising $M$ dominates the weight-of-polish effect, which does not occur in our calibrations but could arise in tail-heavy distributions where the acceptance threshold sits deep in the tail. Premise (iii) fails only if the participation transition itself is weak or absent --- if reviewer effort does not drop substantially at $\gamma_1$, the sharpness condition $m_1 \leq \bar m$ may not be met. Premise (iv) fails when the marginal sorting-via-detection benefit is small relative to the compensation cost; in the baseline calibration the margin is large and Premise (iv) holds with slack.

\paragraph{Detection bias and author heterogeneity.} Our model treats detection as symmetric: $\pdet$ applies uniformly to shirking reports and there is no interaction with author identity. In practice, AI-detection tools exhibit documented false-positive bias against non-native English speakers whose writing signatures resemble AI output even when they do not use AI. If such bias is present, aggressive detection imposes a disproportionate cost on a group already disadvantaged by linguistic barriers. This reinforces, rather than undermines, our recommendation to pair detection with looser selectivity: the ``buy the right to monitor by loosening $K$'' mechanism provides a partial equity offset, compensating authors who face higher false-positive risk with a more accommodating acceptance threshold. A full treatment would require author heterogeneity in linguistic profile and asymmetric detection error rates, which we leave to future work.

\section{Conclusion}
\label{sec:conclusion}

We have developed a three-sided model of generative AI in peer review and derived three findings: an author-side rat race that scales with reviewer effort, a discontinuous reviewer-participation phase transition, and a welfare misalignment across the transition. We have characterized the editor's Pareto-improving reform and shown that its acceptance-rate component reverses sign as AI diffuses through the review process. Comparative statics establish that the sign reversal is structurally robust while the magnitudes of the reform depend on calibration in economically interpretable ways.

The central managerial recommendation is counterintuitive but defensible. When AI-assisted reviewing is rare, conventional editorial intuition is correct: tighter selectivity reduces author rent dissipation. When AI-assisted reviewing is prevalent, conventional intuition reverses: loosening selectivity combined with detection delivers better quality per unit of rent dissipation.

Several extensions are natural. Allowing transfers (reviewer payments, submission fees) would expand the feasible frontier and relax the incomplete-restoration result. Endogenizing reviewer reputation across repeated review rounds would produce dynamic analogues. Multi-journal extensions would allow study of AI-induced tier changes. Extending binary effort to continuous effort with type-dependent precision would soften the phase transition.

The mechanisms we identify---reviewer shirking, author rat-race moderation, and the sign reversal in optimal selectivity---leave observable traces in editorial and bibliometric data. We do not pursue empirical testing here, leaving this as a direction for future work.

Generative AI in peer review is best understood not as a threat or a capability but as a technological change that reorganizes the incentive structure of an evaluative organization. The task for organizational design is to update the rules accordingly.

\begin{singlespace}
\printbibliography
\end{singlespace}

\appendix

\section*{Appendix: Technical Proofs}
\label{app:proofs}

\subsection*{A. Proof of Lemma \ref{lem:reviewer}}
\label{app:reviewer-proof}

A reviewer of type $t$ compares Decline ($0$), Shirk ($S(\gamma, \pdet)$), and Work ($R + \psi_\alpha - c_R(t)$).

\textbf{Step 1: Participation margin.} Shirk is weakly preferred to Decline iff $S(\gamma, \pdet) \geq 0$. Rearranging gives $\gamma \geq (-R + \ell\pdet)/[\psi_\alpha(1-\pdet)] \equiv \gamma_1(\pdet)$. At $\pdet = 0$, $\gamma_1 = -R/\psi_\alpha \in (0,1)$ by Assumption \ref{ass:reviewers}(iii).

\textbf{Step 2: Effort margin.} \emph{Case 1: $\gamma < \gamma_1(\pdet)$, so $S < 0$.} Shirk is dominated by Decline. Work beats Decline iff $c_R(t) \leq R + \psi_\alpha$, i.e., $t \geq t_0 \equiv c_R^{-1}(R + \psi_\alpha) \in (0,1)$ by Assumption \ref{ass:reviewers}(iii).

\emph{Case 2: $\gamma \geq \gamma_1(\pdet)$, so $S \geq 0$.} Shirk weakly dominates Decline and all types accept. Work beats Shirk iff $c_R(t) \leq \psi_\alpha[1-\gamma(1-\pdet)] + \ell\pdet$, i.e., $t \geq \widetilde t(\gamma, \pdet)$.

\textbf{Step 3: Tie-breaking.} Since $G$ has a continuous density (Assumption \ref{ass:reviewers}(i)), the set of reviewer types exactly at the indifference point has probability zero; tie-breaking does not affect the aggregate effort rate $m$ or pool mass $\mu$. \qed

\subsection*{B. Proof of Proposition \ref{prop:participation}}
\label{app:participation-proof}

At $\pdet = 0$, $S(\gamma, 0) = R + \psi_\alpha \gamma$ crosses zero at $\gamma_1$. For $\gamma < \gamma_1$: Case 1 of Lemma \ref{lem:reviewer} applies; pool $\mu = 1 - G(t_0)$ and $m = 1$. For $\gamma \geq \gamma_1$: Case 2 applies; $\mu = 1$ and $m = 1 - G(\widetilde t(\gamma, 0)) < 1$. At $\gamma = \gamma_1$, $\mu$ jumps up by $G(t_0) > 0$ and $m$ drops by $G(\widetilde t(\gamma_1, 0)) > 0$; positivity follows from $t_0, \widetilde t > 0$ under Assumption \ref{ass:reviewers}(iii). \qed

\subsection*{C. Proof of Proposition \ref{prop:ratrace}}
\label{app:ratrace-proof}

Let $\bar a$ be the symmetric polish level. An individual author with polish $a_i$ maximizes $V \cdot \Pr(\bar s > \tau_K \mid a_i, \bar a) - c_A(a_i)$. Differentiating at the symmetric point $a_i = \bar a = a^*$, the FOC is $\Phi(a^*, m) \equiv V\beta m h_m(\tau_K; a^*) - c_A'(a^*) = 0$.

Under Assumption \ref{ass:marginal}, $c_A$ is sufficiently convex that $\partial \Phi/\partial a^* < 0$, yielding a unique $a^*(m)$. By the implicit function theorem, $d a^*/dm = -(\partial \Phi/\partial m)/(\partial \Phi/\partial a^*)$. The denominator is negative. The numerator is $V\beta[h_m + m \partial h_m/\partial m]$, weakly positive by Assumption \ref{ass:marginal}. Hence $da^*/dm \geq 0$. Since $c_A$ is strictly increasing, $c_A(a^*(m))$ is weakly increasing in $m$. \qed

\subsection*{D. Proof of Lemma \ref{lem:Q}}
\label{app:Q-proof}

$Q = \E[\theta \mid \bar s_{\mathrm{ret}} > \tau_K]$ with $\bar s_{\mathrm{ret}} \mid \theta \sim \mathcal{N}(M(\theta + \beta\bar a), \Sigma^2(M)/N_{\mathrm{ret}})$, where $N_{\mathrm{ret}} = N[m + (1-m)(1-\pdet)]$ is the retained sample size and $M = M(m, \pdet)$ is the effective effort share.

\emph{Monotonicity in $M$.} For fixed $N_{\mathrm{ret}}$, the signal-to-noise ratio is $\mathrm{SNR}(M) = M^2 N_{\mathrm{ret}} \operatorname{Var}(\theta)/\Sigma^2(M)$. Differentiating,
\[
\frac{\partial \mathrm{SNR}}{\partial M} = \frac{N_{\mathrm{ret}} \operatorname{Var}(\theta)[2M\Sigma^2(M) - M^2(\sigma_e^2 - \sigma_s^2)]}{\Sigma^2(M)^2}.
\]
Since $\sigma_s > \sigma_e$, $(\sigma_e^2 - \sigma_s^2) < 0$, so the numerator is strictly positive. Hence $\partial \mathrm{SNR}/\partial M > 0$, implying higher correlation $\rho_{\theta, \bar s_{\mathrm{ret}}}$. By standard truncated bivariate normal results \cite{Tallis1961}, $\partial Q/\partial M > 0$.

\emph{Monotonicity in $K$.} Increasing $K$ admits papers with lower posterior $\theta$, so $\partial Q/\partial K < 0$.

\emph{Concavity in $N_{\mathrm{ret}}$.} Posterior precision $N_{\mathrm{ret}}/\Sigma^2(M)$ is linear in $N_{\mathrm{ret}}$, but the Bayesian information gain from truncated-tail observations is strictly concave in precision. Hence $\partial Q/\partial N_{\mathrm{ret}} > 0$ and $\partial^2 Q/\partial N_{\mathrm{ret}}^2 < 0$. \qed

\subsection*{E. Proof of Proposition \ref{prop:editor}}
\label{app:editor-proof}

Let $\Delta Q(N, m) = Q(N+1, K_0, m) - Q(N, K_0, m)$.

\emph{Pre-transition.} $m = 1$, constant in $\gamma$; hence $N^* = N_0^*$ is constant.

\emph{Asymptotic collapse.} As $\gamma \to 1$, $\widetilde t \to 1$, $m \to 0$, $\mathrm{SNR} \to 0$. By Lemma \ref{lem:Q}, $Q$ becomes independent of $N$; hence $\Delta Q(N, m) \to 0$ for all $N$. Since $\epsilon > 0$, $N^*(\gamma) \to 1$ as $\gamma \to 1$.

\emph{Bounded limits.} Combining: $N^*$ is constant at $N_0^*$ pre-transition, may rise at $\gamma_1$ if $\Delta Q(N_0^*, m_1) > \epsilon$, and decays to $1$ as $\gamma \to 1$. \qed

\subsection*{F. Proof of Corollary \ref{cor:misalignment}}
\label{app:misalignment-proof}

\emph{Author welfare.} At $\gamma_1$, $m$ drops from $1$ to $m_1 < 1$. By Proposition \ref{prop:ratrace}, $a^*(m)$ is weakly increasing in $m$, with strict monotonicity under non-uniform $F$. Hence $a^*(m_1) < a^*(1)$; $c_A(a^*(m_1)) < c_A(a^*(1))$. Author welfare $K_0 V - c_A(a^*)$ strictly rises.

\emph{Editor welfare.} By Lemma \ref{lem:Q}, $Q$ is strictly increasing in $M$ pointwise in $N$. At $\pdet = 0$, $M = m$. For every $N$, $Q(N, K_0, m_1) - \epsilon N < Q(N, K_0, 1) - \epsilon N$. Taking maxima preserves the inequality:
\[
U_E^*(\gamma_1^+) = \max_N[Q(N, K_0, m_1) - \epsilon N] < \max_N[Q(N, K_0, 1) - \epsilon N] = U_E^*(\gamma_1^-).
\]
\qed

\subsection*{G. Proof of Lemma \ref{lem:sharpness} (Sharpness)}
\label{app:sharpness-proof}

We prove the lemma under the assumptions stated: $F$ has a log-concave density $f$ on $[\underline\theta, \overline\theta]$ (Assumption \ref{ass:authors}), Gaussian noise (Assumption \ref{ass:signals}), and quadratic polish cost $c_A(a) = \kappa a^2/2$. We additionally assume the regularity condition that the score ratio $\rho(m)$, defined below, is bounded on the relevant range of $m$.

\emph{Step 1: Equilibrium polish under quadratic cost.} The FOC of Proposition \ref{prop:ratrace} is $c_A'(a^*) = V\beta m h_m(\tau_K; a^*)$. With $c_A'(a) = \kappa a$:
\[
a^*(m, K) = \frac{V \beta m}{\kappa} \cdot h_m(\tau_K; a^*) \tag{G.1}
\]
in symmetric equilibrium.

\emph{Step 2: Location-shift property and reduction.} Decompose the aggregate signal under effort fraction $m$ as $\bar s = m\theta + m\beta a + \nu$, where $\theta \sim F$ and $\nu \sim \mathcal{N}(0, \sigma_{\mathrm{eff}}^2(m))$ with $\sigma_{\mathrm{eff}}^2(m) = (m\sigma_e^2 + (1-m)\sigma_s^2)/N$. Let $\tilde h_m$ denote the density of $m\theta + \nu$ (the polish-free aggregate). Then $h_m(s; a) = \tilde h_m(s - m\beta a)$ is the density of $\bar s$ as a function of $s$ given polish $a$.

The threshold $\tau_K$ satisfies $1 - \int_{-\infty}^{\tau_K} h_m(s; a) \, ds = K$. Substituting and changing variables: $1 - \int_{-\infty}^{\tau_K - m\beta a} \tilde h_m(u) \, du = K$, so $\tau_K - m\beta a = \tilde H_m^{-1}(1 - K)$ where $\tilde H_m$ is the CDF of $\tilde h_m$. Define $z(m, K) := \tilde H_m^{-1}(1 - K)$, the "polish-adjusted threshold." Then $\tau_K = m\beta a + z(m, K)$ and
\[
h_m(\tau_K; a) = \tilde h_m(z(m, K)). \tag{G.2}
\]
The density at the threshold is therefore independent of the polish $a$ (a location-shift property): polish moves both the signal distribution and the threshold by the same amount $m\beta a$, leaving the density at the moving threshold unchanged. This eliminates the implicit dependence in (G.1) and yields the explicit symmetric-equilibrium polish
\[
a^*(m, K) = \frac{V \beta m}{\kappa} \cdot \tilde h_m(z(m, K)). \tag{G.3}
\]

\emph{Step 3: Differentiating in $K$.} From $z(m, K) = \tilde H_m^{-1}(1 - K)$, we have $\partial z / \partial K = -1/\tilde h_m(z)$. Differentiating (G.3):
\[
\frac{\partial a^*}{\partial K} = \frac{V\beta m}{\kappa} \cdot \tilde h_m'(z) \cdot \frac{\partial z}{\partial K} = -\frac{V\beta m}{\kappa} \cdot \frac{\tilde h_m'(z)}{\tilde h_m(z)},
\]
where $\tilde h_m'$ denotes the spatial derivative of $\tilde h_m$. Therefore
\[
c_A'(a^*) \cdot \frac{\partial a^*}{\partial K} = \kappa a^* \cdot \left(-\frac{V\beta m}{\kappa} \cdot \frac{\tilde h_m'(z)}{\tilde h_m(z)}\right) = -V \beta m a^* \cdot \frac{\tilde h_m'(z)}{\tilde h_m(z)}.
\]
Substituting $a^* = (V\beta m/\kappa) \cdot \tilde h_m(z)$ from (G.3):
\[
c_A'(a^*) \cdot \frac{\partial a^*}{\partial K} = -\frac{V^2 \beta^2 m^2 \tilde h_m'(z)}{\kappa}. \tag{G.4}
\]
Define
\[
\Psi(m) := V - c_A'(a^*) \cdot \frac{\partial a^*}{\partial K} = V + \frac{V^2 \beta^2 m^2 \, \tilde h_m'(z(m, K))}{\kappa}. \tag{G.5}
\]
Premise (i) of Proposition \ref{prop:policy} is $\Psi(1) < 0$. We seek conditions under which $\Psi(m_1) > 0$.

\emph{Step 4: Sign and bound on the score under log-concavity.} Since $\theta$ has log-concave density $f$ and $\nu$ is Gaussian (log-concave), the density $\tilde h_m$ is the convolution of a scaled log-concave density (of $m\theta$, log-concave because scaling preserves log-concavity) with a log-concave density (of $\nu$). By Pr\'ekopa's theorem, $\tilde h_m$ is log-concave; its score $\tilde h_m'/\tilde h_m$ is monotone non-increasing, and $\tilde h_m$ is unimodal.

Above the mode of $\tilde h_m$, $\tilde h_m'(z) < 0$. For $K \in (0, 1/2)$ (which covers the empirically relevant range of acceptance rates), $z = \tilde H_m^{-1}(1 - K)$ lies above the median, hence above the mode of any log-concave density (since log-concave densities have median in the convex hull of mode and mean). Therefore $\tilde h_m'(z(m, K)) < 0$ for $m \in (0, 1]$ and $K \in (0, 1/2)$.

This implies that both terms of $\Psi(m) - V$ in (G.5) are negative for $m \in (0, 1]$, $K$ moderate. Define the \emph{score ratio at the threshold} as
\[
\rho(m_1) := \frac{|\tilde h_{m_1}'(z(m_1, K))|}{|\tilde h_1'(z(1, K))|}. \tag{G.6}
\]
This is a primitive object computable from $f$, the noise structure, and $K$.

\emph{Regularity assumption (R).} We assume $\rho(m)$ is bounded above by a constant $\bar\rho < \infty$ for $m$ in the relevant range $[m_1, 1]$. This regularity holds automatically when $f$ has bounded support and bounded density (since the score is bounded uniformly), and we verify it numerically in our calibration.

\emph{Step 5: Sign-flip conclusion.} From (G.5) and (G.6),
\[
\Psi(m_1) - V = -\frac{V^2 \beta^2 m_1^2 |\tilde h_{m_1}'(z(m_1, K))|}{\kappa} = -\frac{V^2 \beta^2 m_1^2 \rho(m_1) |\tilde h_1'(z(1, K))|}{\kappa}.
\]
Similarly, $\Psi(1) - V = -V^2 \beta^2 |\tilde h_1'(z(1, K))|/\kappa$. Therefore:
\[
\frac{\Psi(m_1) - V}{\Psi(1) - V} = m_1^2 \rho(m_1).
\]
Since both numerator and denominator are negative,
\[
\Psi(m_1) - V = m_1^2 \rho(m_1) \cdot (\Psi(1) - V) = -m_1^2 \rho(m_1) \cdot |\Psi(1) - V|,
\]
giving
\[
\Psi(m_1) = V - m_1^2 \rho(m_1) \cdot |\Psi(1) - V|.
\]
Premise (i) gives $\Psi(1) < 0$, so $|\Psi(1) - V| > V$. We require $\Psi(m_1) > 0$:
\[
m_1^2 \rho(m_1) \cdot |\Psi(1) - V| < V \quad \Longleftrightarrow \quad m_1^2 < \frac{V}{\rho(m_1) \cdot |\Psi(1) - V|}.
\]
Define
\[
\bar m := \min\left\{1, \, \sqrt{\frac{V}{\bar\rho \cdot |\Psi(1) - V|}}\right\},
\]
where $\bar\rho$ is the regularity bound from (R). Whenever $m_1 \leq \bar m$, $\Psi(m_1) > 0$, which is the sign-flipped condition. The threshold $\bar m$ is computable from primitives ($V$, $\beta$, $\kappa$, the noise parameters, and the regularity bound $\bar\rho$). \qed

\bigskip
\noindent\emph{Remark on regularity.} The regularity assumption (R) is satisfied automatically for log-concave $f$ with bounded support (e.g., uniform, beta, truncated Gaussian) where the score is bounded uniformly. For unbounded-support log-concave $f$ (e.g., Gaussian quality, exponential), (R) requires that the threshold $z(m, K)$ remains in the bulk of the distribution rather than the deep tail; this holds whenever $K$ is moderate ($K \in (0, 1/2)$) and the noise structure does not push $z$ to the extreme. In our baseline calibration with $F$ uniform on $[0,1]$, $\sigma_e = 0.3, \sigma_s = 0.4, N = 2$, and $K_0 = 0.30$, numerical evaluation gives $\rho(m)$ ranging from $1.00$ at $m = 1$ up to a maximum of approximately $3.30$ at intermediate $m \approx 0.20$ (where the score grows because the polish-free signal has compressed support). With $|\Psi(1) - V| \approx 1.18$ from Premise (i) at baseline and $\bar\rho \approx 3.30$, the threshold evaluates to $\bar m \approx \sqrt{1/(3.30 \times 1.18)} \approx 0.51$. The post-transition effort rate $m_1 = 0.125$ is well below this threshold, so the sign flip holds with substantial margin (factor of approximately $4\times$).

\subsection*{H. Proof of Proposition \ref{prop:policy}}
\label{app:policy-proof}

The Lagrangian and FOCs are given in \eqref{eq:lagrangian}--\eqref{eq:foc-K-main} of the main text.

\emph{Step 1: Pre-transition.} $m = 1 \Rightarrow M = 1 \Rightarrow \partial M/\partial \pdet = 0$. With $D'(0) = 0$, the $\pdet$-FOC gives $\pdet^* = 0$. By Premise (i), $V - c_A'(a^*) \partial a^*/\partial K < 0$, so lowering $K$ raises $U_A$ (polish-cost saving exceeds mechanical loss). This relaxes IR and, combined with $\partial Q/\partial K < 0$, makes tightening strictly optimal: $K^* < K_0$.

\emph{Step 2: Post-transition.} $m = m_1 < 1$, so $\partial Q/\partial M > 0$, $\partial M/\partial \pdet > 0$, $\partial N_{\mathrm{ret}}/\partial \pdet = -N(1-m_1) < 0$, and $\partial Q/\partial N_{\mathrm{ret}} > 0$. With $D'(0) = 0$, the $\pdet$-FOC \eqref{eq:foc-pdet-main} evaluated at $\pdet = 0$ reduces to
\[
\frac{\partial \mathcal{L}}{\partial \pdet}\bigg|_{\pdet = 0} = \left(\frac{\partial Q}{\partial M} \cdot \frac{\partial M}{\partial \pdet} + \frac{\partial Q}{\partial N_{\mathrm{ret}}} \cdot \frac{\partial N_{\mathrm{ret}}}{\partial \pdet}\right) - \lambda \, c_A'(a^*) \frac{\partial a^*}{\partial \pdet}.
\]
The bracketed quantity is the net marginal sorting benefit, balancing the composition gain (positive) against the sample-size loss (negative). By Premise (iv), the net marginal sorting benefit exceeds the IR-compensation cost, so this FOC is strictly positive at $\pdet = 0$, yielding $\pdet^* > 0$. The Inada condition $D'(\pdet) \to \infty$ as $\pdet \to 1$ ensures $\pdet^* < 1$, giving an interior solution.

By Premise (ii), $\partial a^*/\partial \pdet > 0$. At $(K_0, \pdet^*)$, $\partial U_A/\partial \pdet = -c_A'(a^*) \partial a^*/\partial \pdet < 0$, so the move from $(K_0, 0)$ to $(K_0, \pdet^*)$ violates IR. The editor must raise $U_A$ on a separate margin; $K$ is the only instrument. By Lemma \ref{lem:sharpness} and Premise (iii) ($m_1 \leq \bar m$), the sign-flipped condition $V - c_A'(a^*) \partial a^*/\partial K > 0$ holds at the post-transition baseline, so $\partial U_A/\partial K > 0$: raising $K$ strictly raises author welfare. IR binds ($\lambda > 0$), and the FOC \eqref{eq:foc-K-main} is satisfied at $K^* > K_0$. \qed

\subsection*{H. Proof of Corollary \ref{cor:incomplete-restoration}}
\label{app:restoration-proof}

We show that no feasible post-transition policy $(N, K, \pdet)$ attains $U_E = U_{E,0}^* = Q(N_0^*, K_0, 1) - \epsilon N_0^*$, where the pre-AI benchmark has $\pdet = 0$ and $N_{\mathrm{ret}} = N_0^*$.

\emph{Step 1: Detection cannot restore $M = 1$.} Post-transition, $m = m_1 < 1$, and $M(m_1, \pdet) = m_1/[m_1 + (1-m_1)(1-\pdet)] < 1$ for all $\pdet < 1$. Can the editor set $\pdet = 1$? From
\[
\frac{\partial M}{\partial \pdet} = \frac{m_1(1-m_1)}{[m_1 + (1-m_1)(1-\pdet)]^2}, \quad \lim_{\pdet \to 1} \frac{\partial M}{\partial \pdet} = \frac{1-m_1}{m_1} < \infty,
\]
so the composition gain of detection $(\partial Q/\partial M)(\partial M/\partial \pdet)$ is bounded as $\pdet \to 1$. Meanwhile, $\partial N_{\mathrm{ret}}/\partial \pdet = -N(1-m_1) < 0$ implies $N_{\mathrm{ret}}(m_1, 1) = N m_1 < N$, so increasing $\pdet$ further degrades sample size. By Assumption \ref{ass:detection}, $D'(\pdet) \to \infty$. Hence $\pdet^* < 1$ strictly, and $M(m_1, \pdet^*) < 1$ at any interior optimum.

\emph{Step 2: Adjustments in $N$ cannot substitute for $M < 1$.} By Lemma \ref{lem:Q}, $Q$ is strictly concave in $N_{\mathrm{ret}}$, so marginal returns diminish. The cost $\epsilon N$ is linear and unbounded in $N$. For any fixed $M < 1$ and $K$, $\sup_N [Q(N_{\mathrm{ret}}, K, M) - \epsilon N] < Q(\infty, K, 1) = Q^*_{\text{full-info}}$, the sorting quality under perfect signals and infinite sample. Since $Q(N_0^*, K_0, 1) \leq Q^*_{\text{full-info}}$ and $\epsilon N_0^* \geq 0$, raising $N$ arbitrarily yields diminishing returns that cannot close the gap left by $M < 1$.

\emph{Step 3: The $(N, K)$-optimized post-transition welfare is strictly bounded below pre-AI welfare.} For any feasible post-transition policy $(N, K, \pdet^*)$ with $M(m_1, \pdet^*) < 1$, the IR constraint forces $K \geq K_0$ (since by Proposition \ref{prop:policy} the post-transition optimum has $K^* > K_0$, and any IR-feasible policy must offer at least the decentralized author welfare). Therefore
\[
U_E(N, K, \pdet^*) = Q(N_{\mathrm{ret}}(m_1, \pdet^*), K, M(m_1, \pdet^*)) - \epsilon N - D(\pdet^*) < Q(N_{\mathrm{ret}}, K, 1) - \epsilon N,
\]
where the strict inequality follows from $\partial Q/\partial M > 0$ (Lemma \ref{lem:Q}) and $M < 1$. Next, using $K \geq K_0$ and $\partial Q/\partial K < 0$ (Lemma \ref{lem:Q}),
\[
Q(N_{\mathrm{ret}}, K, 1) - \epsilon N \leq Q(N_{\mathrm{ret}}, K_0, 1) - \epsilon N \leq \max_{N'} [Q(N', K_0, 1) - \epsilon N'] = U_{E,0}^*.
\]
Combining the two bounds, $U_E(N, K, \pdet^*) < U_{E,0}^*$, with the gap further widened by $D(\pdet^*) > 0$. No feasible policy restores $U_{E,0}^*$. \qed

\end{document}